\documentclass[journal,letterpaper,twoside, web]{ieeecolor}
\usepackage{lcsys}
\usepackage{cite}
\usepackage{amsmath,amssymb,amsfonts}
\usepackage{algorithmic}
\usepackage{graphicx}
\usepackage{textcomp}
\usepackage{comment}
\usepackage[dvipsnames]{xcolor}
\usepackage{subcaption}
\usepackage{float}
\usepackage{url}
\usepackage{tikz}

\let\labelindent\relax
\usepackage{enumitem}

\DeclareMathOperator*{\argmin}{arg\,min}

\newcommand{\At}{\tilde{A}}

\newcommand{\Qt}{\tilde{Q}}

\newcommand{\utp}[2]{u_{#1}^{\pi_{#2}}}

\newcommand{\uogdt}[1]{u^{\mathrm{ogd}}_{#1}}

\makeatletter
\newcommand*{\rom}[1]{\expandafter\@slowromancap\romannumeral #1@}
\makeatother

\newtheorem{theorem}{Theorem}[section]
\newtheorem{lemma}[theorem]{Lemma}

\newtheorem{definition}[theorem]{Definition}

\newtheorem{assumption}{Assumption}

\title{ Online Linear Quadratic Tracking with\\ Regret Guarantees}
\author{Aren Karapetyan, Diego Bolliger, Anastasios Tsiamis, Efe C. Balta,  and John Lygeros
\thanks{This work has been supported by the Swiss National Science Foundation under NCCR Automation (grant agreement $51\text{NF}40\_180545$), and by the  European Research Council under the ERC Advanced grant agreement  $787845$ (OCAL). Aren Karapetyan and Diego Bolliger contributed equally to this work.}
\thanks{Aren Karapetyan, Anastasios Tsiamis, and John Lygeros are with the Automatic Control Laboratory, Department of Information Technology and Electrical Engineering, ETH Z{\"u}rich, 8092 Z{\"u}rich, Switzerland (e-mail: \{akarapetyan, atsiamis, jlygeros\}@control.ee.ethz.ch).}
\thanks{Diego Bolliger is with the School of Engineering, ZHAW Zurich University of Applied Sciences, 8400 Winterthur, Switzerland (e-mail: diego.bolliger@zhaw.ch).}
\thanks{Efe C. Balta  is  with the Control and Automation Group, Inspire AG, 8005 Zürich, Switzerland (e-mail: efe.balta@inspire.ch).}
}

\newcommand{\ak} {\textcolor{black}}
\newcommand{\blue} {\textcolor{black}}

\begin{document}
\maketitle
\thispagestyle{empty}

\newcommand\copyrighttext{%
  \scriptsize \copyright 2023 IEEE. This version has been accepted for publication at the IEEE Control Systems Letters, DOI: 10.1109/LCSYS.2023.3345809. Personal use of this material is permitted.  Permission from IEEE must be obtained for all other uses, in any current or future media, including reprinting/republishing this material for advertising or promotional purposes, creating new collective works, for resale or redistribution to servers or lists, or reuse of any copyrighted component of this work in other works.}
\newcommand\copyrightnotice{%
\begin{tikzpicture}[remember picture,overlay]
\node[anchor=south,yshift=5pt] at (current page.south) {\fbox{\parbox{\dimexpr\textwidth-\fboxsep-\fboxrule\relax}{\copyrighttext}}};
\end{tikzpicture}%
}

\copyrightnotice

\begin{abstract}
Online learning algorithms for dynamical systems provide finite time guarantees for control in the presence of sequentially revealed cost functions. We pose the classical linear quadratic tracking problem in the framework of online optimization where the time-varying reference state is unknown \textit{a priori} and is revealed after the applied control input. We show the equivalence of this problem to the control of linear systems subject to adversarial disturbances and propose a novel online gradient descent-based algorithm to achieve efficient tracking in finite time. We provide a dynamic regret upper bound scaling linearly with the path length of the reference trajectory and a numerical example to corroborate the theoretical guarantees. 
\end{abstract}

\begin{IEEEkeywords}
Optimal Tracking, Online Control.
\end{IEEEkeywords}

\section{Introduction}
\label{sec:introduction}

\IEEEPARstart{L}{inear} quadratic tracking (LQT) is the natural generalization of the optimal linear quadratic regulator (LQR) for the setting where the goal is not to drive the state to the origin but to a certain reference. The reference trajectory need not be necessarily time-invariant and in the classic formulation of the problem is known in advance. This is a reasonable assumption in many practical applications, such as aircraft tracking of a predetermined trajectory or precision control in industrial process engineering. However, in other scenarios, for example, in tracking the output of a secondary agent whose dynamics are unknown and/or the measurements are imperfect, the prediction of the next reference point is non-trivial. In these cases the reference trajectory is only revealed sequentially, after the action has been taken, suggesting the need for an online or adaptive algorithm that will learn or adapt to the dynamics of the reference-generating agent.

In this letter, we study the LQT problem with an unknown reference trajectory. We pose the problem in the framework of online convex optimization (OCO) subject to the dynamics constraint of the system \ak{\cite{hazan2020nonstochastic,foster2020logarithmic}}. In particular, the tracking problem is recast into an equivalent regulation problem with a redefined state that evolves with linear dynamics subject to additive adversarial disturbances. \ak{In the spirit of online decision-making under computational and memory constraints, our goal is to develop a gradient-based algorithm that is fast and simple to implement and requires no large memory.
To this end,} we show how classical online gradient descent (OGD) may fail to achieve optimal tracking and propose a modified algorithm, called SS-OGD (steady state OGD) that is guaranteed to achieve the goal under mild conditions. Given the online nature of the algorithm, its performance is quantified through the means of dynamic regret that compares the accumulated finite time cost of a given algorithm to that of an optimal benchmark that solves the LQT problem with an \emph{a priori} knowledge of the reference trajectory. We provide \ak{a dynamic regret bound that scales linearly with} the path length of the reference trajectory. 

The LQT problem for sequentially revealed adversarial reference states \ak{is studied mostly with policy regret guarantees, with one of the first works \cite{abbasi2014tracking} suggesting a relatively computationally heavy algorithm.} In a more recent line of work \cite{nonhoff2020online} the authors introduce a memory-based, gradient descent algorithm and in \cite{nonhoff2022online} tackle the constrained tracking problem. \ak{Several works also provide dynamic regret guarantees for tracking of unknown targets, however, their settings differ from ours. In \cite{balta2022regret}, the authors analyze an output tracking scheme but assume an iterative setting, while in \cite{li2019online} a window of predictions is available. Without predictions, their regret order is determined by that of a fixed oracle controller.  In \cite{li2019online}, the authors also provide a lower bound for the dynamic regret in terms of the reference path length, matching the same order as our proposed scheme. Gradient-based algorithms, as the ones we study, have also been developed in the context of online feedback optimization \cite{hauswirth2021optimization}. There,  in contrast to our setting, the dynamics are generally assumed to be unknown, but, crucially, the cost functions are fixed over the horizon, and the regret is not analyzed. Several recent works, e.g. \cite{cothren2022online}, consider a similar setting with time-varying costs. These, however, are allowed to be estimated offline by training, incompatible with our setting, and without regret guarantees.}

\textit{Notation}: The set of positive real numbers is denoted by $\mathbb{R}_+$  and  that of non-negative integers by $\mathbb{N}$.  For a matrix $W$ the spectral radius and the spectral norm are denoted by $\rho(W)$, and  $\|W\|$, respectively, and $\lambda_{min}(W)$ denotes its minimum eigenvalue. We define $\lambda_W := \frac{1+\rho(W)}{2}$; one can show that if $\rho(W)<1$, there exists a $c_W \in \mathbb{R}_+$ such that for all $k\geq 1$ $\|W^k\|\leq c_W{\lambda_W}^k$. For a given vector $x$, its  Euclidean norm is denoted by $\|x\|$, and the one weighted by some matrix $Q$  by $\|x\|_Q = \sqrt{x^{\top}Qx}$.  

\section{Problem Statement}
\label{sec:problem_statement}

Consider the discrete-time linear time-invariant (LTI) dynamical system, given by
\begin{equation}
    \label{eq:lti_system}
    x_{t+1} = Ax_t + Bu_t, \quad \forall t \in \mathbb{N}, 
\end{equation}
where $x_t \in \mathbb{R}^n$ and $u_t \in \mathbb{R}^m$ are the state and input vectors respectively, and $A\in \mathbb{R}^{n\times n}, B \in \mathbb{R}^{n\times m}$ are \emph{known} system matrices. The goal of the optimal LQT problem is the tracking of a time-varying signal $r_t \in \mathbb{R}^n$, such that the cost
\begin{equation*}
    \|x_T-r_T\|_{P}^2+\sum_{t=0}^{T-1} \|x_t-r_t\|_Q^2 + \|u_t\|_R^2
\end{equation*}
is minimized for some weighting matrices $Q \in \mathbb{R}^{n\times n}$ and $R\in \mathbb{R}^{m\times m}$, and where  $P\in \mathbb{R}^{n\times n}$ is the solution of the discrete algebraic Riccati equation (DARE)\footnote{\ak{The final cost matrix is taken to be $P$ for simplicity. For other values of the terminal cost matrix the results still hold up to an additional constant} \ak{\cite{foster2020logarithmic}}.}
\begin{equation}
     P = Q + A^{\top}PA - A^{\top}PB(R+B^{\top}P B)^{-1}B^{\top}PA.
     \label{eq:riccati}
 \end{equation}

The LQT problem can be recast into an equivalent LQR formulation \cite{karapetyan2023implications} by considering instead the dynamics
\begin{equation}
    \label{eq:lti_system_noise}
    e_{t+1} = A e_t +B u_t + w_t, \quad \forall t \in \mathbb{N},
\end{equation}
with $e_t := x_t - r_t$ and $w_t := Ar_t - r_{t+1}$ for all $t\in\mathbb{N}$, and the corresponding cost function
\begin{equation}
    \label{eq:cost}
J(e_0,u) := \|e_T\|_{P}^2+\sum_{t=0}^{T-1} \|e_t \|_Q^2 + \|u_t\|_R^2.
\end{equation}
When the reference trajectory $r_t$, $t\in\mathbb{N}$ is known at the initial time a closed form solution for the optimal controller that solves the following optimization problem can be obtained
\begin{subequations}
\label{eq:optimal_offline_benchmark}
\begin{align}
   u^{\star}= &\argmin_u \quad J(e_0,u)\\
     &\text{subject to} \quad \eqref{eq:lti_system_noise} \quad \forall~ 0\leq t<T.
\end{align}
\end{subequations}
This controller, often referred to as the optimal offline noncausal controller, can be represented as  a linear feedback on the current state and the future reference \cite{lewis2012optimal, goel2022power}.

Departing from the classical formulation of tracking control, we assume that the reference signal is \emph{unknown} and is only revealed sequentially after the control input has been applied, similar to the adversarial tracking framework in \cite{abbasi2014tracking}. In particular, for each time step $0\leq t<T$:
\begin{enumerate}[label=\arabic*:]
    \item The state $x_t$ and the reference state $r_t$ are observed,
    \item The agent decides on an input $u_t$,
    \item The environment decides on the next reference $r_{t+1}$, which, in turn, determines $w_t$. The error state then evolves according to \eqref{eq:lti_system_noise}, incurring the following cost for the agent
    \begin{equation}
    \label{eq:online_objective}
    c_t(e_t,u_t):= \|Ae_t + Bu_t +w_t\|_Q^2 + \|u_t\|^2_R.
\end{equation}
\end{enumerate}
Note that the online cost~\eqref{eq:online_objective}, depends on the current input $u_t$ and the unknown disturbance $w_t$, and is therefore unknown to the decision maker at timestep $t$; it is revealed only at time $t+1$, after the input $u_t$ has been applied to the system. Our problem formulation fits the online learning framework, with the extra challenge of inherent dynamics. The goal of the controller is then to minimize the online cumulative cost\footnote{For consistency we require $c_{T-1}(e_{T-1},u_{T-1}):= \|Ae_{T-1} + Bu_{T-1} +w_{T-1}\|_{P}^2 + \|u_{T-1}\|^2_R$. To forego unnecessary cluttering of the notation, the separate treatment of the last timestep is implied implicitly.}
\begin{equation*}
    \sum_{t=0}^{T-1} c_t(e_t, u_t)=J(e_0,u)-\|e_0\|^2_Q.
\end{equation*}
This is the same as the LQR cost without the initial state, implying that the minimizers of both problems coincide.

We quantify the finite-time performance of the algorithm through the means of dynamic regret. Consider a policy $\pi:\mathcal{I}\rightarrow\mathbb{R}^m$, mapping from the available information set, $\mathcal{I}$, to the control input space. Its  dynamic regret, given a disturbance signal $w$,  is defined as
\begin{equation}
\label{eq:dynamic_regret}
    \mathcal{R}^{\pi}(w,e_0) = J(e_0,u^{\pi}) - J(e_0,u^{\star}),
\end{equation}
where $u^{\pi}$ is the input generated by $\pi$ and $u^{\star}$ is given by \eqref{eq:optimal_offline_benchmark}.

We allow the trajectory $r_t, \; t\in\mathbb{N}$ to be arbitrary, as long as it remains bounded. \begin{assumption}[Bounded trajectory]
    \label{ass:boundedness}
     There exists a $\bar{R} \in \mathbb{R}_+$, such that  $\|r_t\|\leq \bar{R}$ for all $t\in \mathbb{N}$.
\end{assumption}
The more abruptly a trajectory changes, the harder it is to achieve good tracking performance, especially if the trajectory is unknown beforehand. To capture  this inherent complexity of the problem \ak{with dynamic regret}, we use the well-established notion of path length \cite{li2019online}, \ak{\cite{zinkevich2003online}}.
\begin{definition}[Path Length]
\label{def:pathlength}
The path length of a reference trajectory $r_{0:T}\in \mathbb{R}^{n(T+1)}$ is $ L(T) = \sum_{t=0}^{T-1}\|\Delta r_t\|$, where $\Delta r_t = r_{t+1}-r_t$.
\end{definition}
For more random and abrupt changes in the trajectory, the path length is higher, and one expects the  performance of an online algorithm to deteriorate. Likewise, an efficient algorithm should improve as path length decreases. This is  captured quantitatively by showing at least a linear dependence of the algorithm's regret on the path length. \blue{One can instead choose the complexity term  to be the path length of the artificial disturbances $w_{0:T-1}\in \mathbb{R}^{n(T)}$. The resulting path length will then decrease the closer the reference dynamics are to the given system in a certain operator norm, but will scale linearly with time in the case of a constant mismatch between the two. Since we only assume bounded references, we allow for potentially random references without any underlying dynamics. Hence, we choose $L(T)$ as our complexity term. }
Under the following standard assumptions \ak{on stabilisability and detectablity\cite{green2012linear}} the LQR problem is well-posed.
\begin{assumption}[LQR is well-posed]
\label{ass:standard}
The system $(A,B)$ is stabilisable, the pair $(Q^{\frac{1}{2}},A)$ is detectable and $R\succ 0$. 
\end{assumption}

\section{The SS-OGD Algorithm}
\label{sec:ss_ogd_algorithm}

We consider a control law of the following form
\begin{equation}
    \label{eq:ogd_input}
    u_t = -Ke_t + v_t, \quad \forall 0\leq t<T,
\end{equation}
where $K = (R+B^{\top}PB)^{-1}B^{\top}PA$ is fixed to the optimal LQR gain, and $v_t$ is a correction term that should account for the unknown disturbances; we will employ online learning techniques to update the latter term. 

We investigate the performance of online gradient descent  based algorithms. Consider the following ``naive" update
\begin{equation}
\label{eq:ogd_update}
    v_{t} = v_{t-1} - \alpha \nabla_{v}c_{t-1}(e_{t-1}, u_{t-1}), %
\end{equation}
where  $v_t$ is updated in the opposite direction of the gradient of the most recent cost. Here $\alpha\in R_+$ is the step size and the recursion starts from some $v_0 \in \mathbb{R}^m$. As the online objective is quadratic, the gradient is available in a closed form and the update can be represented as
$v_{t} = v_{t-1} - 2 \alpha (Ru_{t-1} + B^\top Q e_t)$. For the case of a constant reference signal and an underactuated system, the algorithm can converge to a point that is not necessarily the optimal one with respect to infinite horizon cost minimization. This is due to the greedy behavior of the update that does not take into account future dynamics. In this section, we propose a simple modification to this myopic OGD update \eqref{eq:ogd_update}, called SS-OGD that accounts for this shortcoming.

To motivate the SS-OGD update, we consider the steady state solution of \eqref{eq:lti_system_noise} in closed-loop with the affine control law \eqref{eq:ogd_input} when we fix \ak{$v_i = \bar{v}$ and $r_i = \bar{r}$ for all subsequent timesteps $i\geq t$}. Defining $S := (I-A+BK)^{-1} B$, a closed form solution for the steady state and input is given by~\footnote{Note that $\bar{x}$ and $\bar{u}$ are both defined for a given $\bar{v}$ and $\bar{r}$. The dependence is left for simplicity}
\begin{equation}
\label{eq:steady_state_equation_x_u}
     \Bar{x} = S\bar{v}+SK\bar{r}, \qquad \Bar{u} = (I-KS)(\bar{v}+K\bar{r}).
\end{equation}
One can then find the $\Bar{v}$ which will recover the optimal steady state solution by minimizing the time-averaged infinite horizon steady state cost. \ak{For $\Bar{x}$ and $\Bar{u}$ defined as in \eqref{eq:steady_state_equation_x_u},} this is equivalent to minimizing  
\begin{equation}
    \label{eq:steady_state_program}
\argmin_{\bar{v}}  \{c\left(\Bar{x}-\bar{r},\Bar{u}\right) := \|\bar{x}-\bar{r}\|_Q^2 + \|\bar{u}\|_R^2\},
\end{equation}
whose gradient is given by %
\begin{equation}
   \nabla_{\bar{v}} c\left(\Bar{x}\!-\bar{r},\Bar{u}\right)= 2 \left((I-KS)^\top R \Bar{u}\! +\! S^\top Q (\Bar{x} \!- \bar{r}) \right).
    \label{eq:ss_OGD_ss_solution}
\end{equation}
Since $r$ is, in general not constant, \ak{and the steady state condition is not satisfied} , we suggest a new OGD-like update rule on the bias term $v_t$ that is a modified version of the gradient in~\eqref{eq:ss_OGD_ss_solution}. Specifically, the feedback on the steady state error, $\bar{x}-\bar{r}$, is replaced with the measured error, $x_t-r_t$, and  the steady state input, $\Bar{u}$, with the latest applied input, $u_{t-1}$. This results in the following update, named SS-OGD
\begin{equation}
    \label{eq:ss_ogd_update}
    v_t = v_{t-1} - 2\alpha\left((I-KS)^{\top} R u_{t-1}+ S^\top Qe_t\right).
\end{equation}
\ak{The cost ${c}$ in \eqref{eq:ss_OGD_ss_solution} is defined for the steady state $\bar{x}$ and  input $\bar{u}$ and is thus decoupled from the true online cost $c_t$ in \eqref{eq:online_objective} that reflects the current $e_t$ and $u_t$. These are, in general, not at a steady state and $r_t$ is not constant. Thus, ${c}$ is only an auxiliary, hallucinated cost to construct the update \eqref{eq:ss_ogd_update}.}

\begin{lemma}
\label{lem:strong_convexity}
Under Assumption \ref{ass:standard}, \eqref{eq:steady_state_program} is strictly convex \ak{in $\bar{v}$} for any $K\in \mathbb{R}^{m\times n}$, for which $\rho(A-BK)<1$.
\end{lemma}

\begin{proof}
If the matrix $I-KS$ is singular, there exists a $v\in\mathbb{R}^n$, such that $v=KSv$. Then, for $x = Sv$, at steady state $x = Ax +B(KSv-Kx) = Ax$. Given the detectability condition of the pair $(Q^{\frac{1}{2}},A)$, for any unstable, or marginally stable mode of $A$, the matrix $Q\succ 0$. This ensures that the matrix $S^\top QS+(I-KS)^\top R(I-KS)$ is positive definite, which is equivalent to the strong convexity of \eqref{eq:steady_state_program} .  
\end{proof}

 The modifications from the standard OGD can be interpreted as incorporating the dynamics information in the update rule. As we show in the following, this ensures that in the limit, if the algorithm is stable and the reference signal is constant, the SS-OGD converges to the same point as the solution of the LQR problem minimizing \eqref{eq:cost}. Moreover, through the feedback on the state $e_t$ and input $u_{t-1}$, the update rule \eqref{eq:ss_ogd_update} incorporates a proportional integral (\emph{PI}) control on the measured state. This is demonstrated on a quadrotor control example in Section \ref{sec:numerical_examples}, where, with the inherent integrator dynamics of the quadrotor, the SS-OGD achieves a zero steady state error in tracking a position reference signal with a constant rate of change.
 
 To study the SS-OGD update rule, we introduce the following  evolution of the combined system optimizer dynamics
 \begin{equation}
    \label{eq:combined_system_dynamics}
    z_{t+1} = \Tilde{A}z_t + \Tilde{B}w_t, 
\end{equation}
 where $z_t := [v_t^{\top}~e_t^{\top}]^{\top}$, the matrices $\Tilde{A}\in\mathbb{R}^{p \times p}$ and $\Tilde{B}\in \mathbb{R}^{p \times n}$ are defined in Appendix \ref{sec:appendix_matrices} and $p:=m+n$.

 \begin{assumption}\label{ass:alpha}
    The step size  $\alpha>0$ is such that $\rho(\Tilde{A}) < 1$ .
 \end{assumption}

Since all the variables in $\Tilde{A}$ are known \emph{a priori}, we show that there always exists an $\alpha$ satisfying this assumption and provide a sufficient condition in Appendix \ref{sec:appendix_matrices}.

The following theorem shows that, for a constant  $w_t = \bar{w}$ for all $0\leq t<T$, SS-OGD update \eqref{eq:ss_ogd_update} converges to the solution of 
\begin{equation}
\label{eq:step_wise_ss_min_prob}
        \begin{split}
            (\hat{e}_t,& \hat{v}_t) =~\argmin_{(e,v)}\quad \|e\|_Q^2 + \left\|-Ke+v\right\|_R^2\\
		      &~\text{subject to} \; e = (A-BK)e + Bv + \ak{\underbrace{Ar_t - r_{t+1}}_{w_t},}
        \end{split}
\end{equation}
with $\textstyle{r_{T+1}:=r_T}$. The solution of \eqref{eq:step_wise_ss_min_prob} can be interpreted as the steady state and steady state input that minimize the infinite horizon time-averaged cost \eqref{eq:cost}.
 
\begin{theorem}\label{thm:ss_ogd_convergence}
   Under Assumptions~\ref{ass:standard} and ~\ref{ass:alpha},  if $w_t = \bar{w}$ for all $t\in\mathbb{N}$, the steady state  of \eqref{eq:combined_system_dynamics} coincides with the solution of \eqref{eq:step_wise_ss_min_prob}.
\end{theorem}

The proof of the theorem is provided in Appendix \ref{sec:appendix_steady_state_proof}.  As a corollary,  for a constant signal $r_t=\bar{r}$, the update converges to the solution of \eqref{eq:steady_state_program}. Note that this is not always true for the naive OGD update \eqref{eq:ss_ogd_update}, as its fixed point for a fixed disturbance is not necessarily the same as \eqref{eq:step_wise_ss_min_prob}.

\section{Regret Analysis}
\label{sec:regret_analysis}
To characterize the effectiveness of the algorithm for time-varying signals and to provide finite time guarantees, we analyze its dynamic regret and show that it scales with the path length. The next theorem summarizes this main result.

\begin{theorem}
\label{thm:optimal_offline_regret}
Under Assumptions \ref{ass:boundedness}, \ref{ass:standard} and \ref{ass:alpha}, the dynamic regret of the SS-OGD algorithm scales with the path length
\begin{equation*}
    \mathcal{R}^{\mathrm{SS-OGD}}(w,e_0) \leq \mathcal{O}\left(1+ L(T)\right).
\end{equation*}
\end{theorem}

The proof of the theorem is provided in Section \ref{sec:regret_proof} after some auxiliary results.
\begin{lemma}
\label{lem:cost_difference}
(Cost Difference Lemma \cite{kakade2002approximately}) For any two policies $\pi_1, \pi_2$
\begin{equation*}
    J(e_0,u^{\pi_2}) - J(e_0,u^{\pi_1}) = \sum_{t=0}^{T-1}\mathcal{Q}_t^{\pi_1}(e_t^{\pi_2},u_t^{\pi_2}) - J_t(e_t^{\pi_2},u^{\pi_1}),
\end{equation*}
where $e_t^{\pi_2}$ is the state at time $t$ achieved by applying the policy $\pi_2$, $u_t^{\pi_2}$ is the input generated by the policy $\pi_2$ at time t, $\ak{\mathcal{Q}}_t^{\pi_1}(e,u) = \|e\|_Q^2 + \|u\|_R^2 + J_{t+1}(Ae+Bu+w_t,u^{\pi_1})$ is the Q-function for policy $\pi_1$ and $J_i(e_i,u)$
is the cost-to-go at time step $i$, with initial state $e_i$ and control signal $u$.
\end{lemma}
The proof is omitted, as it is identical to the one for Markov decision processes \cite{kakade2002approximately}. The following result for a general policy $\pi$ akin to the result in \cite{zhang2021regret} follows.

\begin{lemma}
\label{lem:regret_formulation}
   \ak{Under Assumption \ref{ass:standard}}, given the system dynamics \eqref{eq:lti_system_noise} and cost function \eqref{eq:cost}, the dynamic regret of any policy $\pi$ is given by
    \begin{equation*}
      \mathcal{R}^\pi(w, e_0)  = \sum_{t = 0}^{T-1} \left(\utp{t}{} - u_t^{\star}\right)^\top \left(R + B^T P B\right) \left(\utp{t}{} - u_t^{\star}\right),
    \end{equation*}
    where $u_t^\pi$ and $u_t^\star$ denote the inputs generated by  $\pi$, and the optimal policy, both evaluated at the policy state $e_t^\pi$ at time $t$.
\end{lemma}

\begin{proof}
    Let $ \ak{\mathcal{Q}}^{\star}_t (e,u)$ be the optimal Q-function, associated with the optimal control law $u_t^{\star}$. Then, using Lemma \ref{lem:cost_difference} the dynamic regret of the policy $\pi$ is given by
    \begin{equation}
    \label{eq:regret_cost_diff_lemma}
    \textstyle{
        \mathcal{R}^\pi(w, e_0) =\sum_{t = 0}^{T-1} \mathcal{Q}_t^{\star}(e_t^{\pi},u_t^{\pi}) - \min_{u}\ak{\mathcal{Q}}_t^{\star}(e_t^{\pi},u),}
    \end{equation}
    \ak{i.e., a sum of differences of $\mathcal{Q}^\star_t$, evaluated at $u_t^\pi$ and $u^\star$, its minimizer. For an  input, $u\in \mathbb{R}^m$, and some $f\in\mathbb{R}^m$, $g\in \mathbb{R}$} 
        \begin{align*}
        \ak{\mathcal{Q}}^{\star}_t (e,u) &= \|e\|_Q^2 + \|u\|_R^2 + J_{t+1}(Ae+Bu+w_t,u^{\star})\\
        &= u^{\top}(R+B^{\top}PB)u + f^{\top}u+g,
    \end{align*}
    \ak{where the last equality follows from the closed form of $J_{t+1}(x,u^{\star})$ as an extended quadratic function of $x$ \cite{goel2022power, karapetyan2022regret}. Thus, since $u_t^\star$ minimizes an extended quadratic function, 
    $
    \mathcal{Q}_t^{\star}(e_t^{\pi},u_t^{\pi}) -\mathcal{Q}_t^{\star}(e_t^{\pi},u_t^\star) =\|\utp{t}{} - u_t^\star\|^2_{\left(R + B^T P B\right)}$.}
\end{proof}
For future references, we also recall the Cauchy Product inequality defined for two finite series $\{a_i\}_{i=1}^T$ and $\{b_i\}_{i=1}^T$:
    \begin{equation}
        \label{eq:cauchy_product}\textstyle{
        \sum_{i=0}^T\left|\sum_{j=0}^{i}a_jb_{i-j}\right| \leq  \left(\sum_{i=0}^T|a_i|\right) \left(\sum_{j=0}^T|b_j|\right)}.
    \end{equation}
    
\subsection{Proof of Theorem \ref{thm:optimal_offline_regret}}
\label{sec:regret_proof}
 As Lemma \ref{lem:regret_formulation} suggests, the dynamic regret depends on the stepwise control input difference,
    \begin{align*}
        		&\left\| \uogdt{t} - u_t^\star\right\| = \left\| -K e_t + v_t + K e_t + \sum_{i=t}^{T-1} K_w^{i,t} w_i \right\|\\
          &\leq \underbrace{\left\| v_t +\sum_{i=t}^{\infty} K_w^{i,t} w_t \right\|}_{s_{1,t}} + \underbrace{\left\|\sum_{i=t}^{T-1}K_w^{i,t}\Delta w_{i,t}\right\|}_{s_{2,t}} + \underbrace{\left\|\sum_{i=T}^{\infty}K_w^{i,t}w_t\right\|}_{s_{3,t}},
    \end{align*}
    where $\Delta w_{i,t} = w_i - w_t$ and for all $0\leq t \leq i <T $,
    \begin{equation*}
        K_w^{i,t} = (R+B^{\top}PB)^{-1}B^{\top}\left((A-BK)^{\top}\right)^{i-t}P.
    \end{equation*}
    We proceed by bounding each of the above terms separately.
    
\textbf{Term $\boldsymbol{s_{2,t}}$:}  This captures the deviation of the artificial disturbance term from the one fixed at timestep $t$. By noting that $\Delta w_{i,t}$ can be represented as a telescopic sum,
    \begin{align*}
        s_{2,t} &\leq c_F d \sum_{i=t}^{T-1}\lambda_F^{i-t}\sum_{j=t}^{i-1}\|w_{j+1} - w_j\|
        \\
        & \leq \frac{c_Fd}{1-\lambda_F}\sum_{j=t}^{T-2}\|w_{j+1}-w_j\| \lambda_F^{j-t},
    \end{align*}
    where $F:=A-BK$ and  $d = \| (R+B^{\top}PB)^{-1}B^{\top}\|\cdot\|P\|$. Using \eqref{eq:cauchy_product} 
    \begin{align*}
        &\sum_{t=0}^{T-1}s_{2,t} \leq \frac{c_Fd}{1-\lambda_F}\sum_{t=1}^{T-1}\sum_{j=t}^{T-1}\|w_{j}-w_{j-1}\|\lambda_F^{j-t}\\
        &\leq \frac{c_Fd}{1-\lambda_F}\sum_{j=1}^{T-1}\sum_{t=1}^{j}\|w_{j}-w_{j-1}\|\lambda_F^{j-t} \\ &\leq \frac{c_F d}{(1-\lambda_F)^2}\sum_{j=1}^{T-1}\|w_{j}-w_{j-1}\|
        \leq \frac{c_F d\left(\|A\|+1\right)}{(1-\lambda_F)^2}\cdot L(T)
    \end{align*}
    \textbf{Term $\boldsymbol{s_{3,t}}$:} This captures the effect of truncating the infinite horizon problem to a finite one\ak{
    \begin{align*}
        &s_{3,t} \leq c_Fd\left(\|A\|+1\right)\bar{R}\sum_{i=T}^{\infty}\lambda_F^{i-t}\leq \frac{c_Fd\left(\|A\|+1\right)\bar{R}\lambda_F^{T-t}}{1-\lambda_F}\\
        &\sum_{t=0}^{T-1}s_{3,t} \leq  \frac{c_Fd\left(\|A\|+1\right)\Bar{R}\left(1-\lambda_F^T\right)}{(1-\lambda_F)^2},
    \end{align*}
where $\bar{R}$ is defined in Assumption \ref{ass:boundedness}}.

 \textbf{Term $\boldsymbol{s_{1,t}}$:} This captures the cost of performing a gradient step in the direction of the steady state solution instead of the full solution, for a fixed $w_t$. Note that $-\sum_{i=t}^{\infty} K_w^{i,t} w_t$ is the solution of the following infinite horizon optimization problem and is independent of the initial state \cite{goel2022power, yu2020power}
    \begin{equation*}
        	\begin{split}
		\hat{v}_t &=\argmin_{v}\lim_{T\rightarrow \infty} \frac{1}{T}\sum_{i=0}^{T}\|e\|_Q^2 + \|Ke+v\|_R^2\\
  &\begin{split}
		\text{subject to} \quad &e = (A-BK)e +Bv + w_t,
		\label{eq:infinite_horizon_problem}
  \end{split}
	\end{split}
    \end{equation*}
    which is equivalent to \eqref{eq:step_wise_ss_min_prob}. Hence, by Theorem \ref{thm:ss_ogd_convergence}
    \begin{equation*}
        \sum_{i=t}^{\infty} K_w^{i,t} w_t = -\left[I \; 0\right](I-\Tilde{A})^{-1}\Tilde{B}w_t = -\left[I \; 0\right]\hat{z}_t,
    \end{equation*}
where $\hat{z}_t = [ \hat{v}^{\top}_t~\hat{e}^{\top}_t]^{\top}:= (I - \Tilde{A})^{-1}\Tilde{B}w_t$ is the steady state of the SS-OGD dynamics \eqref{eq:combined_system_dynamics} for a given $w_t$. This term captures the difference between the SS-OGD update term $v_t$ and the steady state value $\hat{v}_t$ for that timestep. We look at the evolution of the augmented state difference; for all $0< t\leq T$
\begin{align}
\label{eq:epsilon_evolution}
    \varepsilon_{t} & \ak{ := z_t - \hat{z}_t} = \Tilde{A}z_{t-1} + \Tilde{B}w_{t-1} - \hat{z}_{t}\\
    &= \Tilde{A}z_{t-1} + (I - \Tilde{A})\hat{z}_{t-1} - \hat{z}_{t}\notag =\Tilde{A}\varepsilon_{t-1} + \hat{z}_{t-1} -\hat{z}_{t}. 
\end{align}
Then $ \varepsilon_t = \At^t\varepsilon_0 + \sum_{i=0}^{t-1}\At^i\left(\hat{z}_{t-i-1}-\hat{z}_{t-i}\right)$ for a given time step $0\leq t \leq T$. Under Assumption~\ref{ass:alpha}
    \begin{equation*}
    \begin{split}
        &\|\varepsilon_t\| \leq c_{\At}\lambda_{\At}^t \ak{\|\varepsilon_0\|}\\
        &+ \left\|\left(I - \At\right)^{-1}\Tilde{B}\right\| \cdot \sum_{i=0}^{t-1}\lambda_{\Tilde{A}}^i\left(\|A\|\|\Delta r_{t-i-1}\| + \|\Delta r_{t-i}\|\right).%
        \end{split}
    \end{equation*}
    Defining \ak{$h = \left\|\left(I - \At\right)^{-1}\Tilde{B}\right\|\left(\|A\|+1\right)$,}  $b = (h+1)\bar{R} + \|x_0\| +\|v_0\|$, and $\Bar{\varepsilon} = c_{\Tilde{A}}\left(b+\frac{2\Bar{R}h}{1-\lambda_{\At}}\right)$
    \begin{align}
         \label{eq:epsilon_bound}
         &\|\varepsilon_t\| \leq c_{\Tilde{A}}b\lambda_{\Tilde{A}}^t + c_{\Tilde{A}}h\sum_{i=0}^{t-1}\lambda_{\Tilde{A}}^i\left(\|\Delta r_{t-i}\| + \|\Delta r_{t-1-i}\|\right),\\
         \notag
         &s_{1,t} \leq \|\varepsilon_t\| \leq c_{\Tilde{A}}b\lambda_{\Tilde{A}}^t + c_{\Tilde{A}}h\sum_{i=0}^{t}\lambda_{\Tilde{A}}^i\|\Delta r_{t-i}\|,\\
         \notag
        &\sum_{t=0}^{T-1}s_{1,t}\leq\sum_{t=0}^{T-1}\|\varepsilon_t\| \leq \frac{c_{\Tilde{A}}}{1-\lambda_{\At}}\left(b+hL(T)\right).
    \end{align}
 There exist $s_2,s_3\in\mathbb{R}_+$, such that $s_{2,t}\leq s_2, \; s_{3,t}\leq s_3$ and from  \eqref{eq:epsilon_bound} $s_{1,t}\leq \Bar{\varepsilon}$ for all $t\in\mathbb{N}$. Using Lemma \ref{lem:regret_formulation} and denoting $\Bar{P} = 4\|R+B^{\top}PB\|$
    \begin{align*}
        \mathcal{R}(w,e_0) < \Bar{P}\sum_{t=0}^{T-1}\left(s_{1,t}^2+s_{2,t}^2+s_{3,t}^2\right)\leq {\mathcal{O}\left(1+ L(T)\right)}.
    \end{align*}
\blue{Note that, unlike the regret bound of the FOSS algorithm in \cite{li2019online}, the constant multiplying $L(T)$ above does not depend on $\bar{R}$, but only on system parameters. This implies that the complexity term captures only the relative distance of the references and is not amplified by their upper bounds.}

\subsection{Steady State Benchmark}
Given Theorem \ref{thm:ss_ogd_convergence}, one can also compare  SS-OGD  to the steady state optimal solution for each timestep. Consider
\begin{equation}
 \label{eq:steady_state_controller}
 \hat{u}_t = -K\hat{e}_t+\hat{v}_t   
\end{equation} 
for all $0\leq t< T$, where $\hat{e}_t$ and $\hat{v}_t$ solve \eqref{eq:step_wise_ss_min_prob}. This steady state controller can be interpreted as an optimal benchmark that is decoupled from the system dynamics, has access to the current cost $c_t$, \ak{and hence to the one step ahead reference, $r_{t+1}$,} and solves for its optimal, steady state solution.
The following Lemma provides a side result on the regret of the SS-OGD algorithm with respect to the steady state controller \eqref{eq:steady_state_controller},  $\mathcal{R}_{\mathrm{SS}}^{\mathrm{SS-OGD}}(w,e_0) := J(e_0,u^\mathrm{SS-OGD}) - J(e_0,\hat{u})$.
\begin{lemma}
Under Assumptions \ref{ass:boundedness}, \ref{ass:standard}, and \ref{ass:alpha}, the regret of the SS-OGD algorithm with respect to the steady state benchmark \eqref{eq:steady_state_controller} scales with the reference path length
\begin{equation*}
    \mathcal{R}_{\mathrm{SS}}^{\mathrm{SS-OGD}}(w,e_0) \leq \mathcal{O}\left(1+ L(T)\right).
\end{equation*}
\end{lemma}
\begin{proof}
    The regret can be expressed as a function of the combined error state $\varepsilon$ evolving according to \eqref{eq:epsilon_evolution}. Defining $\Qt_i:=\Qt$ for all $0\leq i<T$ and $\Qt_T$ as in \eqref{eq:Q_tilde} in Appendix \ref{sec:appendix_matrices}
    \begin{equation*}
             \mathcal{R}_{\mathrm{SS}}^{\mathrm{SS-OGD}}(w,e_0) \leq \|\Qt\|\left(2h\Bar{R}+\Bar{\varepsilon}\right)\sum_{t=0}^T\|\varepsilon_t\|,
    \end{equation*}
    \ak{using \eqref{eq:epsilon_bound} and $\|\hat{z}_t\|\leq h\Bar{R},\; \forall t$. Then}
        \begin{equation*}
    \mathcal{R}_{\mathrm{SS}}^{\mathrm{SS-OGD}}(w,e_0)\leq \frac{c_{\Tilde{A}}\|\Qt\|\left(2h\Bar{r}+\Bar{\varepsilon}\right)}{1-\lambda_{\At}}\left(b+hL(T)\right),
    \end{equation*}
    \ak{using \eqref{eq:epsilon_bound}, and  the Cauchy Product inequality \eqref{eq:cauchy_product}.}
\end{proof}

\section{Numerical Example}
\label{sec:numerical_examples}

The SS-OGD algorithm is implemented on a \ak{linearized quadrotor model} \cite{beuchat2019n} in closed-loop with a \emph{PI} velocity controller \cite{Karapetyan_Distributed_Control_of}, to track a reference trajectory in two dimensions.  In particular, we consider the following model 
\begin{align*}
    &\resizebox{!}{1.25cm}{$A = \begin{bmatrix}
		1.000 & 0 & 0.096 & 0 & 0 & 0.040\\
		0 & 1.000 & 0 & 0.096 & -0.040 & 0\\
        0 & 0 & 0.894 & 0 & 0 & 0.703\\
        0 & 0 & 0 & 0.894 & -0.703 & 0\\
        0 & 0 & 0 & 0.193 & 0.452 & 0\\
        0 & 0 & -0.193 & 0 & 0 & 0.452
	\end{bmatrix}$}, \\ 
	&\resizebox{!}{0.55cm}{$B = \begin{bmatrix}
	0.004 & 0 & 0.106 & 0 & 0 & 0.193\\
     0 & 0.004 & 0 & 0.106 & -0.193 & 0
	\end{bmatrix}^{\top}$},
\label{eq:naive_ogd_exp_sysmat}
\end{align*}
where the state $x := \begin{bmatrix}p_x&p_y&v_x&v_y&\beta&\rho\end{bmatrix}^\top$ contains the horizontal position, velocity, the roll and pitch angles of the quadrotor, and the input $u: = \begin{bmatrix}v_x^{t}&v_y^{t}\end{bmatrix}^\top$ sets the target horizontal velocities. We take $Q =\mathop{\rm diag}\left(100,100,1,1,0,0\right)$ and $R=0.1\cdot I$. 

In the first experiment, the drone tracks the shape of the letters \textbf{IFA} for an \textit{a priori} unknown reference with a fixed $\Delta r_t$ for all timesteps. SS-OGD's performance is compared to that of the \ak{causal CE controller that solves for the time-averaged infinite horizon steady state cost by fixing all future references to $r_t$, i.e. $\textstyle{r_{i} = r_t},\; t< i<T$. This is equivalent to solving \eqref{eq:step_wise_ss_min_prob} and fixing $r_{t+1} = r_t$. The CE controller does not have access to $r_{t+1}$, as opposed to the steady state benchmark in \eqref{eq:steady_state_controller}.} The results in Figure \ref{fig:ifa} show that even though the CE controller appears to be tracking the reference better in the $(p_x, p_y)$ plot, the \ak{time plot} reveals that it lags behind the reference trajectory, resulting in around $3$ times higher regret, compared to SS-OGD. As the reference signal has a constant rate of change, the double integrator dynamics of the open loop transfer function from the error to the state, allow SS-OGD to achieve perfect position tracking. \ak{When this is not the case SS-OGD again outperforms the CE controller.}

 \begin{figure*}
 \begin{minipage}[t]{0.68\textwidth}
  \centering\captionsetup{width=\textwidth}
  \includegraphics[width=0.88\linewidth]{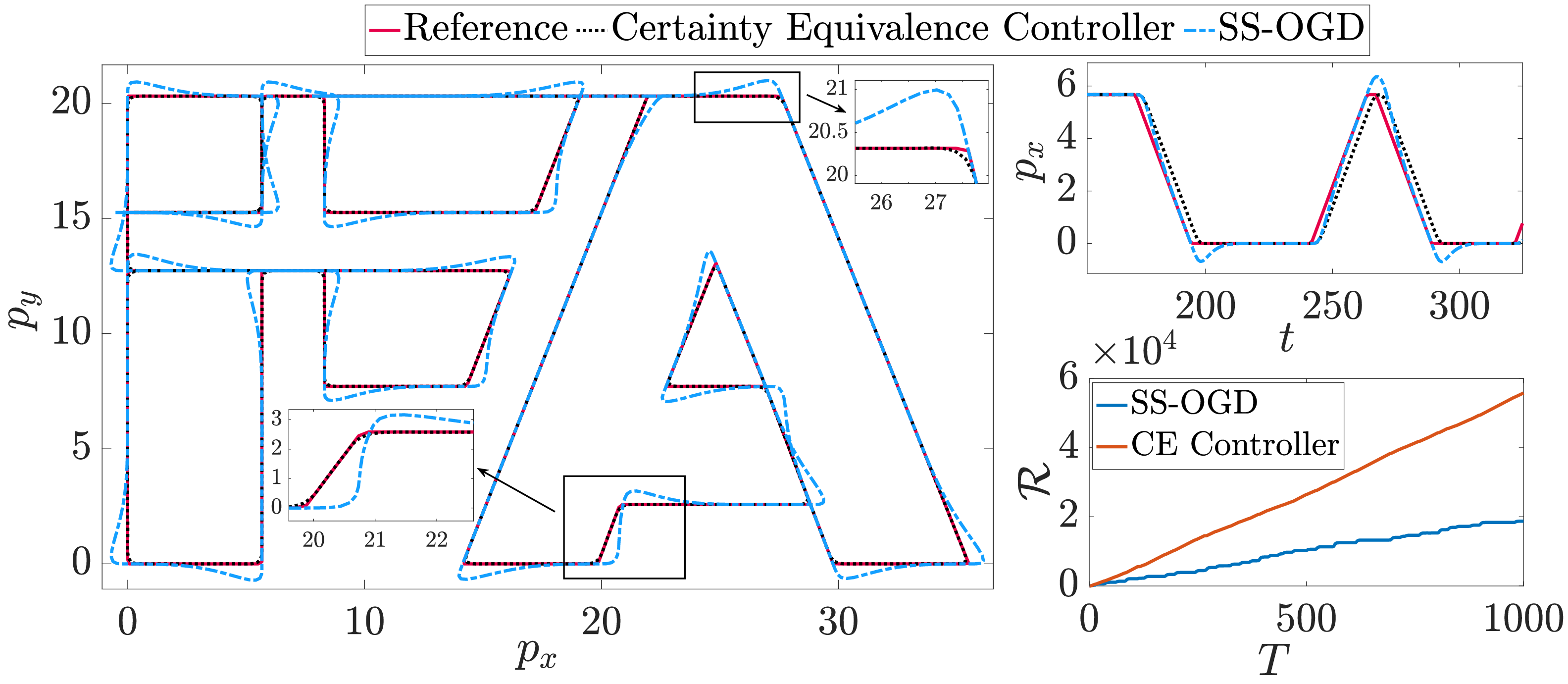}    %
  \caption{Tracking a 2-D shape with a quadrotor model. The horizontal position plot (left panel) \ak{shows the apparent better tracking of the CE controller}. However, the time plot (top right panel) shows \ak{its} visible time lag; by contrast SS-OGD quickly converges to the reference. This leads to a lower rate of regret for SS-OGD (bottom right panel).}
  \label{fig:ifa}
\end{minipage}%
\begin{minipage}[t]{0.02\textwidth}
\hfill
\end{minipage}
\begin{minipage}[t]{.28\textwidth}
  \centering\captionsetup{width=\textwidth}
  \includegraphics[width=0.88\linewidth]{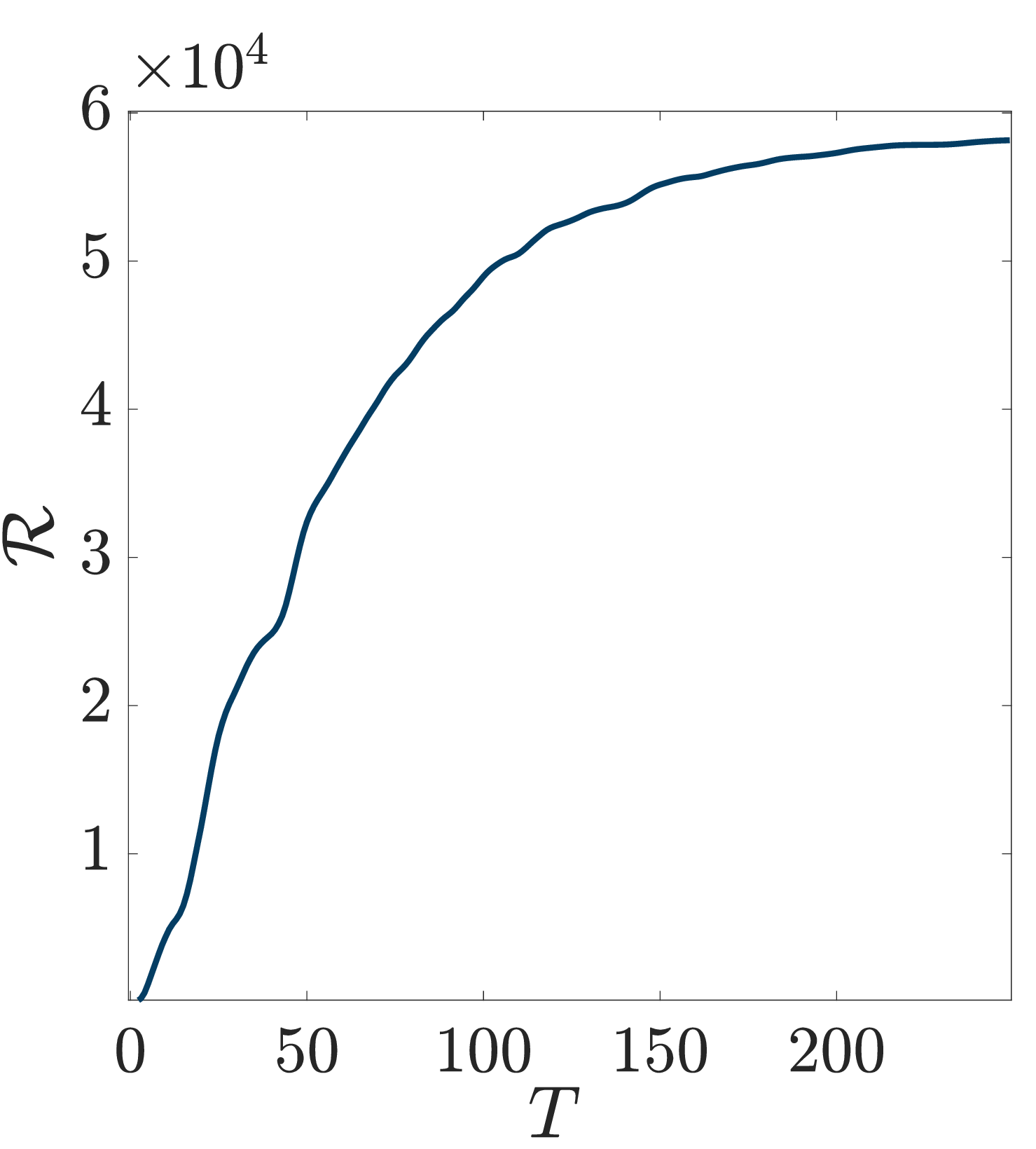}    %
 \caption{\blue{Empirical regret of SS-OGD with a finite reference path length converges to a finite value, as expected from the theoretical bound.}} 
\label{fig:regret}
\end{minipage}
\end{figure*}
In a second experiment  the empirical worst-case regret as a function $T$ \ak{is calculated}. For each $T$, $60$ random reference signals are simulated and the highest value of regret is noted. The references are generated such that $\|\Delta r_t\|$ decreases with a constant factor of $0.99$. This ensures a finite path length \ak{and therefore a finite regret} in the limit, as \ak{shown} in Figure \ref{fig:regret}, and in agreement with the upper bound in Theorem \ref{thm:optimal_offline_regret}.
\section{Conclusion}
In this letter, we reformulate the \ak{online} LQT problem as an online control problem subject to adversarial disturbances. Within this framework, we propose a novel online gradient descent-based algorithm, called SS-OGD, and show that its dynamic regret scales with the path length of the reference signal. We validate the results on numerical examples with a quadrotor model. \blue{The improvement of the regret coefficients, as well as the case where the references are generated by some unknown dynamics is left to be studied in future work.}

\appendices
\section{System-Optimizer Dynamics}
\label{sec:appendix_matrices}
The combined system-optimizer dynamics matrices are 
\begin{equation*}
\Tilde{A} = \begin{bmatrix}
			I-\alpha M & -\alpha H \\
			B & A-BK
\end{bmatrix}, \qquad
\Tilde{B} = \begin{bmatrix}
			-2 \alpha S^T Q \\ I
	\end{bmatrix},
\end{equation*}
where $M: = 2 \left(S^T Q B \ak{+} (I-KS)^\top R\right)$, and $H:=2(S^T Q (A-BK)-(I-KS)^\top R K)$. The objective function in \eqref{eq:step_wise_ss_min_prob} can be equivalently written as $z_t^{\top}\Tilde{Q}z_t$ for $0<t<T$ and as  $z_T^{\top}\Tilde{Q}_Tz_T$ for $t=T$, where
\begin{equation}
\label{eq:Q_tilde}
    \Tilde{Q} \!=\! \begin{bmatrix}
		R & -RK \\ -K^{\top}R & Q\!+\!K^{\top}RK
 	\end{bmatrix}\!,~
      \Tilde{Q}_T \!=\! \begin{bmatrix}
		\boldsymbol{0}_{m \times m} & \boldsymbol{0}_{m \times n} \\ \boldsymbol{0}_{n \times m} & P
 	\end{bmatrix}\!.
\end{equation}

Consider the  coordinate transformation  $\Tilde{A}_V := V\Tilde{A}V^{-1}$
\begin{equation*}
    V = \begin{bmatrix}
        I & \boldsymbol{0}_{m \times n}\\
        -S & I
    \end{bmatrix}, \Tilde{A}_V = \begin{bmatrix}
        I -\alpha\overline{M} & -\alpha H\\
        \alpha S\overline{M} & \alpha SH + (A-BK)
    \end{bmatrix},
\end{equation*}
with $\overline{M} := M+HS = 2 \left(S^{\top}QS + (I-KS)^{\top}R(I-KS)\right)$  positive definite, as shown in Lemma \ref{lem:strong_convexity}. Using the small gain theorem for interconnected systems \cite{green2012linear}, the following, along with $\alpha<2/\rho(\overline{M})$ is a sufficient condition for the stability of $\Tilde{A}_V$ and therefore of the  dynamics \eqref{eq:combined_system_dynamics}
\begin{equation*}
    \alpha\cdot\frac{\|S\overline{M}\|\|H\|}{\lambda_{min}(\overline{M})}\cdot \max_{w\in\mathbb{R}}\big\|\big[e^{jw}I - \alpha SH - (A-BK)\big]^{-1}\big\|<1.
\end{equation*}
Since $A-BK$ is stable, there always exists an arbitrarily small $\alpha>0$  such that the above is fulfilled.
\section{Proof of Theorem~\ref{thm:ss_ogd_convergence}}
\label{sec:appendix_steady_state_proof}
Given a disturbance vector $w_t$ and a bias input $v$, the steady state  of the dynamics \eqref{eq:lti_system_noise} with the control law \eqref{eq:ogd_input} is given by $e = Sv + \hat{S}w_t$, where $\hat{S} := (I -A+BK)^{-1}$.
Substituting this in the objective function of \eqref{eq:step_wise_ss_min_prob}, one can confirm that the $v$ that minimizes that cost is the unique (as shown in Lemma \ref{lem:strong_convexity}) solution of $    \left(S^{\top} QS + \ak{(I-KS)^T  R(I-KS)}\right) v = \left( (I-KS)^{\top} R K \hat{S} - S^{\top} Q \hat{S}\right) w_t$. Using the definition of $\Tilde{A}$, the steady state $\hat{v}$ of the SS-OGD update~\eqref{eq:ss_ogd_update} for a constant $w_t$ \ak{solves $ 0= M\hat{v} + H(S\hat{v}+\hat{S}w_t)+2 S^{\top}Qw_t$}. Then
\begin{align*}
    0 &= S^T Q (\ak{I}+(A-BK)\hat{S})B \hat{v} +(I-KS)^T R(I-KS) \hat{v} \\ 
			&  + S^T Q (I+(A-BK) \hat{S}) w_t-(I-KS)^T R K \ak{\hat{S}w_t}.
\end{align*}
Since $I+(A-BK) \hat{S} = \hat{S}$, \ak{and $S= \hat{S}B$}, the two equations coincide, leading to the unique steady state solution $\hat{v}$.
\bibliographystyle{ieeetr}
\bibliography{bibliography}

\end{document}